\documentclass[preprint,12pt]{elsarticle}

\usepackage{amssymb, amsthm, amsmath}
\usepackage{algorithm, algpseudocode}
\usepackage{url}

\theoremstyle{theorem}
\newtheorem{theorem}{Theorem}
\newtheorem{lemma}{Lemma}

\theoremstyle{definition}
\newtheorem*{definition}{Definition}

\journal{Theoretical Computer Science A}

\begin{document}

\begin{frontmatter}

\title{Linear Time Approximation Schemes for Geometric Maximum Coverage\tnoteref{conference}}
\tnotetext[conference]{The conference version of this paper was published in COCOON 2015: The 21st Annual International Computing and Combinatorics Conference.}
\author[HKU]{Kai Jin\fnref{KJ}}
\author[IIIS]{Jian Li\corref{LJ}}
\author[Utah]{Haitao Wang\fnref{W}}
\author[IIIS]{Bowei Zhang\fnref{ZB}}
\author[IIIS]{Ningye Zhang\fnref{ZN}}
\fntext[KJ]{~\url{cskaijin@hku.hk}}
\cortext[LJ]{corresponding author:~\url{lijian83@mail.tsinghua.edu.cn}}
\fntext[W]{~\url{haitao.wang@usu.edu}}
\fntext[ZB]{~\url{bw-zh14@mails.tsinghua.edu.cn}}
\fntext[ZN]{~\url{zhangny12@mails.tsinghua.edu.cn}}
\address[HKU]{Department of Computer Science, University of Hong Kong, Hong Kong SAR}
\address[IIIS]{Institute for Interdisciplinary Information Sciences (IIIS), Tsinghua University, Beijing, China, 100084\fnref{SupportIIIS}}
\address[Utah]{Department of Computer Science,	Utah State University, Utah, USA, 84322\fnref{SupportUtah}}
\fntext[SupportIIIS]{J. Li, B. Zhang and N. Zhang's research was supported in part by the National Basic Research Program of China Grant 2015CB358700, 2011CBA00300, 2011CBA00301, the National Natural Science Foundation of China Grant 61202009, 61033001, 61361136003. }
\fntext[SupportUtah]{H. Wang's research was supported in part by NSF under Grant CCF-1317143.}

\begin{abstract}
We study approximation algorithms for the following geometric version of the maximum coverage problem:
Let $\mathcal{P}$ be a set of $n$ weighted points in the plane.
Let $D$ represent a planar object, such as a rectangle, or a disk.
We want to place $m$ copies of $D$ such that the sum of the weights of the points in $\mathcal{P}$ covered by these copies is maximized.
For any fixed $\varepsilon>0$, we present efficient approximation schemes
that can find a $(1-\varepsilon)$-approximation to the optimal solution.
In particular, for $m=1$ and for the special case where $D$ is 
a rectangle,
our algorithm runs in time $O(n\log (\frac{1}{\varepsilon}))$, improving 
on the previous result.
For $m>1$ and the rectangular case, our algorithm runs in
$O(\frac{n}{\varepsilon}\log (\frac{1}{\varepsilon})+\frac{m}{\varepsilon}\log m +m(\frac{1}{\varepsilon})^{O(\min(\sqrt{m},\frac{1}{\varepsilon}))})$ time.
For a more general class of shapes (including disks, polygons with $O(1)$ edges), our algorithm runs in $O(n(\frac{1}{\varepsilon})^{O(1)}+\frac{m}{\epsilon}\log m + m(\frac{1}{\varepsilon})^{O(\min(m,\frac{1}{\varepsilon^2}))})$ time.
\end{abstract}


\end{frontmatter}

\newcommand{\calP}{\mathcal{P}}
\newcommand{\calL}{\mathcal{L}}
\newcommand{\calA}{\mathcal{A}}
\newcommand{\calU}{\mathcal{U}}

\newcommand{\MaxRS}{\mathsf{MaxCov_R}}
\newcommand{\MaxC}{\mathsf{MaxCov_C}}
\newcommand{\MaxCRS}{\mathsf{MaxCov_D}}
\newcommand{\Cover}{\mathsf{Cover}}
\newcommand{\OPT}{\mathsf{Opt}}
\newcommand{\R}{\mathbb{R}}
\newcommand{\F}{\mathsf{F}}
\newcommand{\D}{\mathsf{D}}
\newcommand{\RD}{\mathsf{RD}}
\newcommand{\tF}{\widehat{\mathsf{F}}}

\newcommand{\cell}{\mathsf{c}}
\newcommand{\upperd}{\mathsf{b}}
\newcommand{\greedy}{\mathsf{Greedy}}

\newcommand{\UA}{U^\star_\calA}
\newcommand{\UP}{U^\star_\calP}
\newcommand{\Ac}{\calA_\cell}
\newcommand{\Pc}{\calP_\cell}
\newcommand{\eat}[1]{}

\newcommand{\BlockPartion}{\textsc{MaxCovCell}}
\newcommand{\GeneralmPartion}{\textsc{MaxCovCellM}}

\newcommand{\DP}{\textsc{DP}}
\newcommand{\GREEDYP}{\textsc{Greedy}}

\newcommand{\Partition}{\textsc{Partition}}
\newcommand{\topic}[1]{\noindent{\underline{#1}:}}

%% \linenumbers

\section{Introduction}

The maximum coverage problem is a classic problem in theoretical computer science and combinatorial optimization.
In this problem, we are given a universe $\calP$ of weighted elements, a family of subsets and a number $m$.
 The goal is to select at most $m$ of these subsets such that the sum of the weights of the covered elements in $\calP$ is maximized.
It is well-known that the most natural greedy algorithm achieves
an approximation factor of $1-1/e$, which is essentially optimal (unless P=NP) \cite{hochbaum1998analysis,nemhauser1978analysis,feige1998threshold}.
However, for several geometric versions of the maximum coverage problem,
better approximation ratios can be achieved (we will mention some of such results below).
In this paper, we mainly consider the following geometric maximum coverage problem:

\begin{definition}
($\MaxRS(\calP,m)$)
Let $\calP$ be a set of $n$ points in a 2-dimensional Euclidean plane $\R^2$.
Each point  $p\in\calP$  has a given weight $w_{p} \geq 0$.
The goal of our geometric max-coverage problem (denoted as $\MaxRS(\calP,m)$) is to place $m$ $a \times b$ rectangles
such that the sum of the weights of the covered points by these rectangles is maximized.
More precisely, let $S$ be the union of $m$ rectangles we placed. Our goal is to maximize
$$
\Cover(\calP,S)=\sum_{p \in \calP\cap S}{w_{p}}.
$$
\end{definition}

We also study the same coverage problem with other shapes, instead of rectangles.
We denote the corresponding problem for circular disk as $\MaxC(\calP, m)$,
and denote the corresponding problem for general object $\D$ as $\MaxCRS(\calP, m)$.
One natural application of the geometric maximum coverage problem is the facility placement problem.
In this problem, we would like to locate a certain number of facilities to serve the maximum number of clients.
Each facility can serve a region (depending on whether the metric is $L_1$ or $L_2$,
the region is either a square or a disk).

\subsection{$m=1$}

\topic{Previous Results}
We first consider $\MaxRS(\calP,1)$, i.e., 
the maximum coverage problem with 1 rectangle.
Imai and Asano~\cite{imai1983finding}, Nandy and Bhattacharya~\cite{nandy1995unified} gave two different exact algorithms for  computing $\MaxRS(\calP,1)$, both running in time $O(n\log n)$.
It is also known that solving $\MaxRS(\calP,1)$ exactly in algebraic decision tree model requires $\Omega(n\log n)$ time \cite{ben1983lower}.
Tao et al.~\cite{tao2013approximate} proposed a randomized approximation scheme for $\MaxRS(\calP,1)$.
With probability $1-1/n$, their algorithm returns a ($1-\varepsilon$)-approximate answer in
$O(n\log (\frac{1}{\varepsilon})+n\log\log n)$ time.
In the same paper, they also studied the problem in the external memory model.

\medskip
\topic{Our Results}
For $\MaxRS(\calP,1)$ we show that there is an approximation scheme
that produces a ($1-\varepsilon$)-approximation and runs in $O(n\log (\frac{1}{\varepsilon}))$ time,
improving the result by Tao et al.~\cite{tao2013approximate}.

\subsection{General $m>1$}

\topic{Previous Results}
Both $\MaxRS(\calP, m)$ and $\MaxC(\calP, m)$ are NP-hard if $m$ is part of the input \cite{megiddo1984complexity}.
The most related work is de Berg, Cabello and Har-Peled~\cite{de2009covering}.
They mainly focused on using unit disks (i.e., $\MaxC(\calP, m)$).
They proposed a $(1-\varepsilon)$-approximation algorithm for $\MaxC(\calP,m)$
with time complexity
$O(n(m/\varepsilon)^{O(\sqrt{m})})$.
\footnote{
They were mainly interested in the case where $m$ is a constant.
So the running time becomes
$O(n(1/\varepsilon)^{O(\sqrt{m})})$
(which is the bound claimed in their paper)
and
the exponential dependency on $m$ does not look too bad for $m=O(1)$.
Since we consider the more general case,
we make the dependency on $m$ explicit.
}
We note that their algorithm can be easily extended to $\MaxRS$ with the same time complexity.

We are not aware of any explicit result for $\MaxRS(\calP,m)$ for general $m>1$.
It is known \cite{de2009covering} that the problem admits a PTAS  via the standard shifting technique \cite{hochbaum1985approximation}.
\footnote{
Hochbaum and Maass \cite{hochbaum1985approximation}
obtained a PTAS for the problem of covering given points
with a minimal number of rectangles. Their algorithm can be easily modified into a PTAS
for $\MaxRS(\calP,m)$ with running time $n^{O(1/\epsilon)}$.
}

\medskip
\topic{Our Results}
Our main result is an approximation scheme for $\MaxRS(\calP,m)$ which runs in time
$$
O\left( \frac{n}{\varepsilon}\log \frac{1}{\varepsilon}+
        \frac{m}{\varepsilon}\log m + m\left(\frac{1}{\varepsilon}\right)^{\Delta_1}\right),
$$
where $\Delta_1=O(\min(\sqrt{m},\frac{1}{\varepsilon}))$.
Our algorithm can also be extended to other shapes subject to some common assumptions, including disks, polygons with $O(1)$ edges (see Section~\ref{sec:othershape} for the assumptions).
The running time of our algorithm is
$$
O\left(n\Bigl(\frac{1}{\varepsilon}\Bigr)^{O(1)}+
        \frac{m}{\varepsilon}\log m + m\Bigl(\frac{1}{\varepsilon}\Bigr)^{\Delta_2}\right),
$$
where $\Delta_2=O(\min(m,\frac{1}{\varepsilon^2}))$.

Following the convention of approximation algorithms,
$\varepsilon$ is a fixed constant.
Hence, the second and last term is essentially $O(m\log m)$ and
the overall running time is essentially linear $O(n)$ 
(if $m=O(n/\log n)$).

Our algorithm follows the standard shifting technique~\cite{hochbaum1985approximation},
which reduces the problem to a smaller problem restricted in a constant size cell.
The same technique is also used in de Berg et al.~\cite{de2009covering}.
They proceeded by first solving the problem exactly in each cell, and then
use dynamic programming to find the optimal
allocation for all cells.
\footnote{
	In fact, their dynamic programming runs in time at least $\Omega(m^2)$.
	Since they focused on constant $m$, this term is negligible in their running time.
	But if $m>\sqrt{n}$, the term can not be ignored and may become the dominating term.
	}

Our improvement comes from another two simple yet useful ideas.
First, we apply the shifting technique in a different way and make the side length of grids much smaller ($O(\frac{1}{\varepsilon})$, instead of $O(m)$ in de Berg et al.'s algorithm~\cite{de2009covering}). Second, we solve the dynamic program approximately. In fact, we show that
a simple greedy strategy (along with some additional observations) can be used for this purpose, which allows us to save another $O(m)$ term.

\subsection{Other Related Work}

There are many different variants for this problem.
We mention some most related problems here.

Barequet et al. \cite{barequet1997translating}, Dickerson and Scharstein \cite{dickerson1998optimal} studied the max-enclosing polygon problem which aims to find a position of a given polygon to cover maximum number of points.
This is the same as $\MaxRS(\calP,1)$ if a polygon is a rectangle. Imai et al.~\cite{imai1983finding} gave an optimal algorithm for the max-enclosing rectangle problem with time complexity $O(n\log n)$.

$\MaxC(\calP, m)$ was introduced by Drezner~\cite{drezner1981note}.
Chazelle and Lee~\cite{chazelle1986circle} gave an $O(n^{2})$-time exact algorithm for the problem $\MaxC(\calP,1)$.
A Monte-Carlo $(1-\varepsilon)$-approximation algorithm for $\MaxC(\calP,1)$
was shown in \cite{agarwal2002translating}, where $\calP$ is
an unweighted point set.
Aronov and Har-Peled~\cite{aronov2008approximating} showed that for unweighted point sets an $O(n\varepsilon^{-2}\log n)$ time Monte-Carlo $(1-\varepsilon)$-approximation algorithm exists,
and also provided some results for other shapes. de Berg et al.~\cite{de2009covering} provided an $O(n\varepsilon^{-3})$ time $(1-\varepsilon)$-approximation algorithm.

For $m>1$, $\MaxC(\calP,m)$ has only a few results.
For $m=2$, Cabello et al.~\cite{cabello2008covering} gave an exact algorithm for
this problem when the two disks are disjoint in $O(n^{8/3}\log^{2}n)$ time. de Berg et al.~\cite{de2009covering} gave $(1-\varepsilon)$-approximation algorithms
that run in $O(n\varepsilon^{-4m+4}\log^{2m-1}{(1/\varepsilon)})$ time for $m>3$ and in $O(n\varepsilon^{-6m+6}\log{(1/\varepsilon)})$ time for $m=2,3$.

The dual of the maximum coverage problem is the classical set
cover problem. The geometric set cover problem
has enjoyed extensive study in the past two decades.
The literature is too vast to list exhaustively here.
See e.g., \cite{Bronnimann,Clarkson,even2005hitting,mustafa2009ptas,varadarajan2010weighted,Chan2012,li2015ptas} and the references therein.

\paragraph*{Outline}
We consider the rectangular case first, and then show the extension to general shapes in the last section.

\section{Preliminaries} % Major section
\label{sec:prel}

We first define some notations and mention some results that are needed in our algorithm.
Denote by $G_{\delta}(a,b)$ the square grid with mesh size $\delta$
such that the vertical and horizontal lines  are defined as follows
\begin{eqnarray*}
G_{\delta}(a,b)=& \left\{(x,y)\in\mathbb{R}^{2}\mid y=b+k\cdot\delta,k\in \mathbb{Z}\right\}
\\ & \cup\left\{(x,y)\in\mathbb{R}^{2}\mid x=a+k\cdot\delta,k\in \mathbb{Z}\right\}.
\end{eqnarray*}
Given $G_{\delta}(a,b)$ and a point $p=(x,y)$, we call the integer pair
$(\lfloor x/\delta\rfloor,\lfloor y/\delta\rfloor)$ the {\em index} of $p$
(the index of the cell in which $p$ lies in).

%------------------------------------------------
\medskip
\topic{Perfect Hashing}
Dietzfetbinger et al. \cite{dietzfelbinger1997reliable}
shows that if each basic algebraic operation (including $\{+,-,\times,\div,\log_2,\exp_2\}$) can be done in constant time,
we can get a perfect hash family so that each insertion and membership query takes $O(1)$ expected time.
In particular, using this hashing scheme,
we can hash the indices of all points, so that we can obtain the list of all non-empty cells in $O(n)$ expected time.
Moreover, for any non-empty cell, we can retrieve all points lies in it in time linear in the number of such points.

\medskip
\topic{Linear Time Weighted Median and Selection} % Sub-section
It is well known that finding the weighted median for an array of numbers can be done in deterministic
worst-case linear time.
The setting is as follows:
Given $n$ distinct elements $x_{1},x_{2},...,x_{n}$ with positive weights $w_{1},w_{2},...,w_{n}$.
Let $w=\sum_{i=1}^{n}w_{i}$.
The {\em weighted median} is the element $x_{k}$
satisfying $\sum_{x_{i}<x_{k}}w_{i}<w/2$ and $\sum_{x_{i}>x_{k}}w_{i}\leq w/2$.
Finding the $k$-th smallest elements for any array can also be done in deterministic
worst-case linear time. See e.g., \cite{clrs}.

\medskip
\topic{An Exact Algorithm for $\MaxRS(\calP, 1)$} As we mentioned, Nandy and Bhattacharya \cite{nandy1995unified} provided an $O(n\log n)$ exact algorithm for the $\MaxRS(\calP,1)$ problem. We use this algorithm as a subroutine in our algorithm.
%Roughly speaking, We can think about the area where a rectangle can cover a fixed point. This area it self is a rectangle.
%So we can use a horizontal line to sweep the plane. We maintain a segment tree on the line, adding and remove a point to the segment tree according to the vertical order of the possible covering area. Then the maximum value we ever get in the segment tree is exactly the answer we want. We are not going to present the algorithm precisely here, see[2] for more detail.
%------------------------------------------------
%----------------------------------------------------------------------------------------
%	MAJOR SECTION X - TEMPLATE - UNCOMMENT AND FILL IN
%----------------------------------------------------------------------------------------

\section{A Linear Time Algorithm for $\MaxRS(\calP,1)$}
\label{sec:m1}

\textbf{Notations:} Without loss of generality, we can assume that $a=b=1$, i.e., all the rectangles are $1\times1$ squares,
(by properly scaling the input).
%We only consider the case of $a=b=1$,i.e. unit squares for the rest of the paper.
We also assume that all points are in general positions.
In particular, all coordinates of all points are distinct.
For a unit square $r$, we use $w(r)$
to denote the sum of the weights of the points covered by $r$.
We say a unit square $r$ is located at $(x,y)$ if the top-left corner of $r$ is $(x,y)$.

Now we present our approximation algorithm for $\MaxRS(\calP,1)$.

\subsection{Grid Shifting} % Sub-section
Recall the definition of a grid $G_{\delta}(a,b)$ (in Section~\ref{sec:prel}).
Consider the following four grids: $G_{2}(0,0)$, $G_{2}(0,1)$, $G_{2}(1,0)$, $G_{2}(1,1)$ with $\delta = 2$.
We can easily see that for any unit square $r$, there exists one of the above grids that does not intersect $r$
(i.e., $r$ is inside some cell of the grid).
This is also the case for the optimal solution.
%Thus, we get the following shifting algorithm.

Now, we describe the overall framework, which is similar to that in~\cite{tao2013approximate}.
Our algorithm differs in several details.
\BlockPartion($\cell$) is a subroutine that takes a $2\times2$ cell $\cell$ as input and returns a unit square $r$
that is a (1-$\varepsilon$)-approximate solution if the problem is restricted to cell $\cell$.
We present the details of \BlockPartion \ in the next subsection.

\begin{algorithm}[h]
  \caption{$\MaxRS(\calP, 1)$}
  \begin{algorithmic}[]
  \State $w_{\max}\leftarrow0$
  \For {each $G\in\{G_{2}(0,0),G_{2}(0,1),G_{2}(1,0),G_{2}(1,1)\}$}
  \State Use perfect hashing to find all the non-empty cells of $G$.
     \For {each non-empty cell $\cell$ of $G$}
     \State $r\leftarrow$ \BlockPartion($\cell$).
     \State {\bf If} $w(r)> w_{\max}$, {\bf then} $w_{\max}\leftarrow w(r)$ and $r_{\max}\leftarrow r$.
     \EndFor;
  \EndFor;
  \State \textbf{return} $r_{\max}$;
  \end{algorithmic}
\label{algo:mainalgo1}
\end{algorithm}

As we argued above, there exists a grid $G$ such that the optimal solution is inside some cell $\cell^\star \in G$.
Therefore, $\BlockPartion(\cell^\star)$ should return a (1-$\varepsilon$)-approximation for the original problem $\MaxRS(\calP,1)$.

\subsection{\BlockPartion} % Sub-section
\label{subsec:partition}

In this section, we present the details of the subroutine \BlockPartion.
Now we are dealing with the problem restricted to a single $2\times2$ cell $\cell$.
%We first find out the number of points inside this block.
Denote the number of point in $\cell$ by $n_{\cell}$, and the sum of the weights of points
in $\cell$ by $W_{\cell}$.
We distinguish two cases, depending on
whether $n_{\cell}$ is larger or smaller than $\left(\frac{1}{\varepsilon}\right)^{2}$.
If $n_{\cell}<\left(\frac{1}{\varepsilon}\right)^{2}$, we simply apply the $O(n\log n)$ time exact algorithm. \cite{nandy1995unified}

The other case requires more work.
In this case, we further partition cell $\cell$ into many smaller cells.
First, we need the following simple lemma.

\begin{lemma}
\label{lm:partition}
Given $n$ points in $\R^2$ with positive weights $w_{1},w_{2},...,w_{n},$ $\sum_{i=1}^{n}w_{i}=w$.
Assume that $x_{1},x_{2},...,x_{n}$ are their distinct $x$-coordinates.
We are also given a value $w_{d}$ such that $\max(w_{1},w_{2},...,w_{n})\leq w_{d} \leq w$,
Then, we can find at most $2w/w_{d}$ vertical lines such that the sum of the weights of points strictly between (we do not count the points on these lines) any two adjacent lines is at most $w_{d}$
in time $O(n\log(w/w_{d}))$.
\end{lemma}

\begin{algorithm}[h]
  \caption{\Partition($\{x_{1},x_{2},...,x_{n}\}$)}
  \begin{algorithmic}[]
  \State Find the weighted median $x_{k}$ (w.r.t. $w$-weight);
  \State $\calL=\calL\cup\{x_{k}\}$;
  \State Generate $S=\{x_{i}\mid w_{i}<x_{k}\}$, $L=\{x_{i}\mid w_{i}>x_{k}\}$;
  \State If the sum of the weights of the points in $S$ is lager than $w_{d}$, run \Partition(S);
  \State If the sum of the weights of the points in $L$ is lager than $w_{d}$, run \Partition(L);
  \end{algorithmic}
  \label{algo:partition}
\end{algorithm}

\begin{proof}
See Algorithm~\ref{algo:partition}.
In this algorithm, we apply the weighted median algorithm recursively.
Initially we have a global variable $\calL=\emptyset$, which upon termination is
the set of $x$-coordinates of the selected vertical lines.
Each time we find the weighted median $x_{k}$ and separate the point with the vertical line $x=x_{k}$,
which we add into $\calL$.
The sum of the weights of points in either side is at most half of the sum of the weights of all the points.
Hence, the depth of the recursion is at most $\lceil\log(w/w_{d})\rceil$.
Thus, the size of $\calL$ is at most $2^{\lceil\log(w/w_{d})\rceil}\leq 2w/w_{d}$,
and the running time is $O(n\log(w/w_{d}))$.
\end{proof}

We describe how to partition cell $\cell$ into smaller cells.
First, we partition $\cell$ with some vertical lines.
Let $\calL_v$ denote a set of vertical lines. Initially, $\calL_v=\emptyset$.
Let $w_{d}=\frac{\varepsilon\cdot W_{\cell}}{16}$.
We find all the points whose weights are at least $w_d$.
For each such point, the vertical line that passes through this point is added to $\calL_v$.
Then, we apply Algorithm~\ref{algo:partition}
to all the points with weights less than $w_d$.
Next, we add a set $\calL_h$ of horizontal lines in exactly the same way.

\begin{lemma}
The sum of the weights of points strictly between any two adjacent lines in $\calL_v$ is at most $w_d=\frac{\varepsilon\cdot W_{\cell}}{16}$.
The number of vertical lines in $\calL_v$ is at most $\frac{32}{\varepsilon}$.
Both statements hold for $\calL_h$ as well.
\end{lemma}
\begin{proof}
%Obviously, the process of our algorithm guarantees the sum of the weights of points
%between any two adjacent lines is at most $\frac{\varepsilon\cdot W{c}}{16}$.
The first statement is straightforward from the description of the algorithm.
We only need to prove the upper bound of the number of the vertical lines.
Assume the sum of the weights of those points considered in the first (resp. second) step is $W_{1}$(resp. $W_{2}$), $W_{1}+W_{2}=W_{\cell}$.
The number of vertical lines in $\calL_v$ is at most
$$
W_{1}/\left(\frac{\varepsilon\cdot W_{\cell}}{16}\right)+
2W_{2}/\left(\frac{\varepsilon\cdot W_{\cell}}{16}\right)
\leq \frac{32}{\varepsilon}.
$$
The first term is due to the fact that the weight of each point we found in the first step has weight at least $\frac{\varepsilon\cdot W_{\cell}}{16}$, and
the second term directly follows from Lemma~\ref{lm:partition}.
\end{proof}

We add both vertical boundaries of cell $\cell$ into $\calL_v$
and both horizontal boundaries of cell $\cell$ into $\calL_h$.
Now $\calL=\calL_v\cup \calL_h$
forms a grid of size at most $(\frac{32}{\varepsilon}+2) \times (\frac{32}{\varepsilon}+2)$.
Assume
$\calL=\{(x,y)\in\mathbb{R}^{2} \mid y=y_{j},j\in \{1,...,v\}\}\cup\{(x,y)\in\mathbb{R}^{2}\mid x=x_{i},i\in \{1,...,u\}\}$,
with both $\{y_{i}\}$ and $\{x_{i}\}$ are sorted.
$\calL$ partitions $\cell$ into {\em small cells}.
The final step of our algorithm is simply enumerating all the unit squares located at $(x_{i},y_{j}),i\in \{1,...,u\},j\in \{1,...,v\}$,
and return the one with the maximum coverage.
However, computing the coverage exactly for all these unit squares is expensive.
Instead, we only calculate the weight of these unit square approximately as follows.
For each unit square $r$, we only count the weight of points that are in some small cell fully covered by $r$.
Now, we show this can be done in $O\left(n_{\cell}\log \left(\frac{1}{\varepsilon}\right)+\left(\frac{1}{\varepsilon}\right)^{2}\right)$ time.

After sorting $\{y_{i}\}$ and $\{x_{i}\}$, we can use binary search to identify which small cell each point lies in.
So we can calculate the sum of the weights of points at the interior, edges or corners of all small cells
in  $O(n_{\cell}\log \left(\frac{1}{\varepsilon}\right))$ times.

Thus searching the unit square with the maximum (approximate) coverage
can be done with a standard incremental algorithm in $O\left(\frac{1}{\varepsilon}\right)^{2}$ time.

Putting everything together,
we conclude that if $n_{\cell}\geq\left(\frac{1}{\varepsilon}\right)^{2}$,
the running time of \BlockPartion($\cell$)\ is
$O\left(n_c\log \left(\frac{1}{\varepsilon}\right)+\left(\frac{1}{\varepsilon}\right)^{2}\right).$

\begin{lemma}
\label{singlecell}
The subroutine \BlockPartion($\cell$) returns a (1-$\varepsilon$)-approximation to $\MaxRS(\calP_\cell, 1)$, where
$\calP_\cell$ is the set of points in $\calP$ that lies in $\cell$.
\end{lemma}

\begin{proof}
The case $n_{\cell}<\left(\frac{1}{\varepsilon}\right)^{2}$ is trivial since we apply the exact algorithm.
So we only need to prove the case of $n_{\cell}\geq\left(\frac{1}{\varepsilon}\right)^{2}$.

\begin{figure}[h]
\centering
\includegraphics[width=0.35\textwidth]{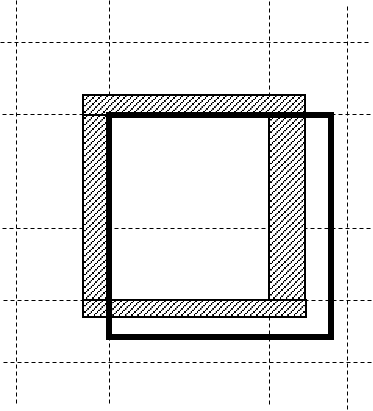}
\caption{Proof of Lemma~\ref{singlecell}.}
\label{fig:1}
\end{figure}

Suppose the optimal unit square is $r$. Denote by $\OPT$ the weight of the optimal solution.
The size of $\cell$ is $2\times 2$, so we can use $4$ unit squares to cover the entire cell.
Therefore, $\OPT\geq \frac{W_{\cell}}{4}$.
Suppose $r$ is located at a point $p$, which is in the strict interior of a small cell $B$ separated by $\calL$.
\footnote{If $p$ lies on the boundary of $B$, the same argument still works.}
Suppose the index of $B$ is $(i,j)$.
We compare the weight of $r$ with $I(i,j)$ (which is the approximate weight of the unit square located at the top-left corner of $B$).
See Figure~\ref{fig:1}.
By the rule of our partition, the weight difference is at most $4$ times the maximum possible weight of points between two adjacent lines in $\calL$.
So $I(i,j)\geq \OPT-4\cdot \frac{\varepsilon\cdot W_{\cell}}{16}\geq (1-\varepsilon)\OPT$.
This proves the approximation guarantee of the subroutine.
\end{proof}

We conclude the main result of this section with the following theorem.

\begin{theorem}
	\label{thm:main1}
Algorithm~\ref{algo:mainalgo1} returns a (1-$\varepsilon$)-approximation to
$\MaxRS(\calP, 1)$
in $O(n\log \left(\frac{1}{\varepsilon}\right))$ time.
\end{theorem}

\begin{proof}
The correctness follows from Lemma~\ref{singlecell} and the previous discussion.

The analysis of the running time is given below. The running time consists of two parts:
 cells with number of points more than $\left(\frac{1}{\varepsilon}\right)^{2}$ and
 cells with number of points less than $\left(\frac{1}{\varepsilon}\right)^{2}$.
 Let $n_{1}\geq n_{2}\geq,...,\geq n_{j}\geq \left(\frac{1}{\varepsilon}\right)^{2}>n_{j+1}\geq n_{j+2},...,\geq n_{j+k}$
 be the sorted sequence of the number of points in all cells. Then, we have that
\begin{align*}
\text{Running time} \leq & \sum_{i=1}^{j}O\left(n_{i}\log \left(\frac{1}{\varepsilon}\right)+\left(\frac{1}{\varepsilon}\right)^{2}\right)+\sum_{i=1}^{k}O\left(n_{i+j}\log(n_{i+j})\right) \\
=
& O\left(\log \left(\frac{1}{\varepsilon}\right)\sum_{i=1}^{j}n_{i}+j\left(\frac{1}{\varepsilon}\right)^{2}+\sum_{i=1}^{k}n_{i+j}\log(n_{i+j})\right) \\
\leq & O\left(\log \left(\frac{1}{\varepsilon}\right)\sum_{i=1}^{j}(n_{i})+n+\sum_{i=1}^{k}n_{i+j}\log\left(\frac{1}{\varepsilon}\right)\right) \\
=    & O\left(\log \left(\frac{1}{\varepsilon}\right)\sum_{i=1}^{j+k}(n_{i})+n\right)
=O\left(n\log \left(\frac{1}{\varepsilon}\right)\right).
\end{align*}
\end{proof}

\section{Linear Time Algorithms for $\MaxRS(\calP,m)$} % Major section
\label{sec:generalm}

For general $m$, we need the shifting technique \cite{hochbaum1985approximation}.

\subsection{Grid Shifting}
Consider grids with a different side length $\frac{6}{\varepsilon}$.
We shift the grid to $\frac{6}{\varepsilon}$ different positions: $(0,0),(1,1),....,(\frac{6}{\varepsilon}-1,\frac{6}{\varepsilon}-1)$.
(For simplicity, we assume that $\frac{1}{\varepsilon}$ is an integer and no point in $\calP$ has an integer coordinate, so points in $\calP$ will never lie on the grid line.
Let $$
\mathbb{G}=\left\{G_{6/\varepsilon}(0,0),...,G_{6/\varepsilon}(6/\varepsilon-1,6/\varepsilon-1)\right\}.
$$
The following lemma is quite standard.
\begin{lemma}
\label{shifting}
There exist $G^\star\in\mathbb{G}$ and a $(1-\frac{2\varepsilon}{3})$-approximate solution $R$
such that none of the unit squares in $R$ intersects $G^\star$.
\end{lemma}

\begin{figure}[t]
\centering
\includegraphics[width=0.65\textwidth]{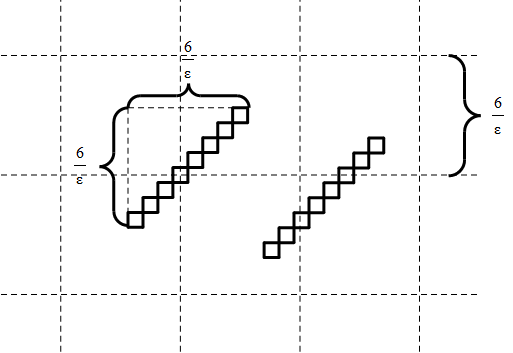}
\caption{Proof of Lemma~\ref{shifting}: the shifting technique.}
\label{trace}
\end{figure}

\begin{proof}
	For any point $p$, we can always use four unit squares to cover the $2\times2$ square centered at $p$.
	Therefore, there exists an optimal solution $\mathsf{OPT}$ such that each covered point is cover by at most 4 unit squares in $\mathsf{OPT}$.
	For each grid $G_{\frac{6}{\varepsilon}}(i,i)\in\mathbb{G}$, we build a modified answer $R_{i}$ from $\mathsf{OPT}$ in the following way. For each square $r$ that intersects with $G_{\frac{6}{\varepsilon}}(i,i)$, there are two different situations. If $r$ only intersects with one vertical line or one horizontal line. We move the square to one side of the line with bigger weight. In this case we will lose at most half of the weight of $r$. Notice that this kind of squares can only intersect with two grids in $\mathbb{G}$.
	Similarly, If $r$ intersects with one vertical line and one horizontal line at the same time, we move it to one of the four quadrants derived by these two lines. In this case we will lose at most 3/4 of the weight of $r$. This kind of squares can only intersect with one grid in $\mathbb{G}$. (see Figure~\ref{trace})
	Now we calculate the sum of the weights we lose from $R_{0},R_{1},...,R_{\frac{6}{\varepsilon}-1}$, which is at most $\max\{1/2\times2,3/4\times1\}=1$ times the sum of weights of squares in $\mathsf{OPT}$.
By the definition of $\mathsf{OPT}$, it is at most $4w(\mathsf{OPT})$.
So the sum of the weights of $R_{0},R_{1},...,R_{\frac{6}{\varepsilon}-1}$ is at least
	$(\frac{6}{\varepsilon}-4)w(\mathsf{OPT})$. Therefore there exists some $i$ such that $R_{i}$(which does not intersect $G_{\frac{6}{\varepsilon}}(i,i)$) is a $(1-\frac{2\varepsilon}{3})$ approximate answer.
\end{proof}

We will approximately solve the problem for each grid $G$ in $\mathbb{G}$ (that is, find an approximation to $R_G$, where $R_G$ denotes the best solution where no squares in $R_G$ intersect $G$), and then select the optimal solution among them.

The idea to solve a fixed grid is as follows. First, we present a subroutine in Subsection~\ref{subsec:F} which can approximately solve the problem for a fixed cell. Then, we apply it to all the nonempty cells.
To compute our final output from those obtained solutions, we apply a dynamic programming algorithm or a greedy algorithm which are shown in the next two sections.

\subsection{Dynamic Programming}
\label{subsec:dp}

Now consider a fixed grid $G\in \mathbb{G}$.
Let $\cell_{1}, \ldots , \cell_{t}$ be the nonempty cells of grid $G$
and $\OPT$ be the optimal solution that does not intersect $G$.
Obviously, $(\frac{6}{\varepsilon})^{2}$ unit squares are enough to cover an entire $\frac{6}{\varepsilon}\times\frac{6}{\varepsilon}$ cell.
Thus the maximum number of unit squares we need to place in one single cell is $m_{c}=\min\{m,(\frac{6}{\varepsilon})^{2}\}$.

Let $\OPT(\cell_i,k)$ be the maximum weight we can cover with $k$ unit squares in cell $\cell_i$.
For each nonempty cell $\cell_i$ and for each $k\in[m_{c}]$, we find a $(1-\frac{\varepsilon}{3})$-approximation $\F(\cell_i,k)$ to $\OPT(\cell_i,k)$. We will show how to achieve this later in Subsection~\ref{subsec:F}. Now assume that we can do it.

Let $\OPT_{\F}$ be the optimal solution we can get from the values $\F(\cell_i,k)$.
More precisely,
\begin{align}
\label{eq:opt}
\OPT_{\F}=\max_{k_1,\ldots,k_t \in [m_{c}]}
\left\{\sum_{i=1}^t \F(\cell_i,k_i)
	\,\,\Big\rvert\,\, \sum_{i=1}^{t} k_i=m\right\}.
\end{align}
We can see that
$\OPT_{\F}$ must be a $(1-\frac{\varepsilon}{3})$-approximation to $\OPT$.
We can easily use dynamic programming to calculate
the exact value of $\OPT_{\F}$.
Denote by $A(i,k)$ the maximum weight we can cover with $k$ unit squares in cells $\cell_{1},\cell_{2},...,\cell_i$.
We have the following DP recursion:
$$
A(i,k)=\left\{
\begin{array}{lcl}
\max_{j=0}^{\min(k,m_{c})}
\left\{A(i-1,k-j)+\F(\cell_i,j)\right\}  &      &  {\text{if} \quad i>1}\\
\F(c_{1},k)  &      &  {\text{if} \quad i=1}\\
\end{array}\right.
$$

The running time of the above simple dynamic programming is $O(m\cdot t\cdot m_{c})$.
One may notice that each step of the DP is computing a $(+,\max)$ convolution.
However, existing algorithms (see e.g., \cite{bremner2006necklaces,williams2014faster}) only run slightly better than quadratic time. So the improvement would be quite marginal.
 But in the next section, we show that if we would like to settle for an
 approximation to $\OPT_{\F}$, the running time can be dramatically improved to linear.

\subsection{A Greedy Algorithm}
\label{subsec:greedy}
We first apply our $\MaxRS(\calP, 1)$ algorithm in Section~\ref{sec:m1}
to each cell $\cell_i$,
to compute a $(1-\frac{\varepsilon^{2}}{9})$-approximation of
 $\OPT(\cell_i,1)$. Let $f(\cell_i,1)$ be the return values.
\footnote{
Both $f(\cell_i,1)$ and $\F(\cell_i,1)$ are approximations of $\OPT(\cell_i,1)$,
with slightly different approximation ratios.
}
This takes $O(n\log{\frac{1}{\varepsilon}})$ time.
Then, we use the selection algorithm to find out
the $m$ cells with the largest $f(\cell_i,1)$ values.
Assume that those cells are $\cell_{1}, ..., \cell_{m}, \cell_{m+1},...,\cell_{t}$,
sorted from largest to smallest by $f(\cell_i,1)$.

\begin{lemma}
\label{selectm}
Let $\OPT(m)$ be the maximum weight we can cover using $m$ unit squares in $\cell_{1}, ..., \cell_{m}$. Then
$\OPT(m)\geq(1-\frac{\varepsilon^{2}}{9})\OPT$.
\end{lemma}

\begin{proof}
Let $k$ be the number of unit squares in $\OPT$ that are chosen from $\cell_{m+1},\ldots,\cell_{t}$. This means there must be at least $k$ cells in $\{\cell_{1},\ldots,\cell_{m}\}$ such that $\OPT$ does not place any unit square.
Therefore we can always move all $k$ unit squares placed in $\cell_{m+1},\ldots,\cell_{t}$ to these empty cells such that each empty cell contains only one unit square. Denote the weight of this modified solution by $A$. Obviously, $\OPT(m)\geq A$.
For any
$i$,$j$ such that $1\leq i\leq m<j\leq t$, we have $\OPT(\cell_i,1)\geq f(\cell_i,1)\geq f(\cell_j,1)\geq (1-\frac{\varepsilon^{2}}{9})\OPT(\cell_j,1)$. Combining with a simple observation that
$\OPT(\cell_j, k)\leq k\OPT(\cell_j, 1)$,
we can see that $A\geq(1-\frac{\varepsilon^{2}}{9})\OPT$. Therefore, $\OPT(m)\geq(1-\frac{\varepsilon^{2}}{9})\OPT$.
\end{proof}

Hence, from now on, we only need to consider the first $m$ cells
$\{\cell_{1},...,\cell_{m}\}$.

Let $\OPT_{\F}(m)$ be the optimal solution we can get from the values $\F(\cell_i,k)$ of the first $m$ cells.
More precisely,
\begin{align}
\label{eq:opt}
\OPT_{\F}(m)=\max_{k_1,\ldots,k_m \in [m_{c}]}
\left\{\sum_{i=1}^m \F(\cell_i,k_i)
	\,\,\Big\rvert\,\, \sum_{i=1}^{m} k_i=m\right\}.
\end{align}

We distinguish two cases.
If $m\leq324(\frac{1}{\varepsilon})^{4}$,
we just apply the dynamic program to compute $\OPT_\F(m)$.
The running time of the above dynamic programming is $O((\frac{1}{\varepsilon})^{O(1)})$.
If $m>324(\frac{1}{\varepsilon})^{4}$,
we can use a greedy algorithm to find an answer of weight at least $(1-\frac{\varepsilon^{2}}{9})\OPT_{\F}(m)$.

Let $\upperd=(\frac{6}{\varepsilon})^2$.
For each cell $\cell_{i}$,
we find the upper convex hull of 2D points
$\{(0,\F(\cell_{i},0))$,$(1,\F(\cell_{i},1))$, \ldots , $(\upperd,\F(\cell_{i},\upperd))\}$.
See Figure~\ref{fig:fandtf}.
Suppose the convex hull points are
$\{(t_{i,0}, \F(\cell_{i},t_{i,0}))$, $(t_{i,1},\F(\cell_{i},t_{i,1}))$, ... , $(t_{i,s_{i}},\F(\cell_{i},t_{i,s_{i}}))\}$, where
$t_{i,0}=0$,$t_{i,s_{i}}=\upperd$.
For each cell, since the above points are already sorted from left to right, we can compute the convex hull in $O(\upperd)$ time by Graham's scan\cite{graham1972efficient}. Therefore, computing the convex hulls for all these cells takes $O(m \upperd)$ time.

\begin{figure}[t]
	\centering
	\includegraphics[width=0.7\textwidth]{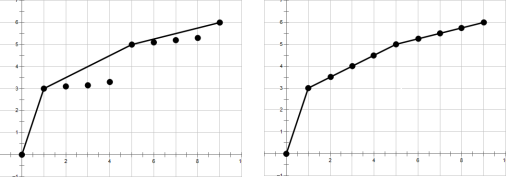}
	\caption{$\F(\cell_{i},k)$ (left) and $\tF(\cell_{i},k)$ (right)}
	\label{fig:fandtf}
\end{figure}

For each cell $\cell_{i}$, we maintain a value $p_{i}$ representing that we are going to place $t_{i,p_{i}}$ squares in cell $\cell_{i}$.
Initially for all $i\in [m]$, $p_{i}=0$.
In each stage, we find the cell $\cell_{i}$ such that
current slope (the slope of the next convex hull edge)
$$
\frac{\F(\cell_{i},t_{i,p_{i}+1})-\F(\cell_{i},t_{i,p_{i}})}{t_{i,p_{i}+1}-t_{i,p_{i}}}$$
is maximized.
Then we add 1 to $p_{i}$, or equivalently
we assign $t_{i,p_{i}+1}-t_{i,p_{i}}$ more squares into cell $\cell_{i}$.
We repeat this step until we have already placed at least $m-\upperd$ squares.
We can always achieve this since we can place at most
$\upperd$ squares in one single cell in each iteration.
Let $m'$ the number of squares we have placed ($m=\upperd\leq m'\leq m$).
For the remaining $m-m'$ squares, we allocate them arbitrarily.
We denote the algorithm by $\GREEDYP$ and let the value obtained be $\greedy(m')$.
Having the convex hulls, the running time of the greedy algorithm is $O(m\log m)$.

Now we analyze the performance of the greedy algorithm.

\begin{lemma}
\label{lm:greedyratio}
The above greedy algorithm computes an $(1-\varepsilon^2/9)$-approximation to
$\OPT_{\F}(m)$.
\end{lemma}
\begin{proof}
Define an auxiliary function $\tF(\cell_{i},k)$ as follows:
If $k=t_{i,j}$ for some $j$, $\tF(\cell_{i},k)=\F(\cell_{i},k)$.
Otherwise, suppose $t_{i,j}<k<t_{i,j+1}$,
then $$
\tF(\cell_{i},k)=
\F(\cell_{i},t_{i,j})+
\frac{\F(\cell_{i},t_{i,j+1})-\F(\cell_{i},t_{i,j})}{t_{i,j+1}-t_{i,j}}\times(k-t_{i,j}).
$$
Intuitively speaking, $\tF(\cell_{i},k)$(See Figure~\ref{fig:fandtf}) is the function defined by the
upper convex hull at integer points.
\footnote{
	At first sight, it may appear that $\F(\cell_i,k)$ should be a concave function.
	However, this is not true. A counter-example is provided in the appendix.}
Thus, for all $i\in[m]$, $\tF(\cell_{i},k)$ is a concave function.
Obviously, $\tF(\cell_i,k)\geq \F(\cell_i,k)$ for all $i\in [m]$ and all $k\in [\upperd]$.

Let $\OPT_{\tF}(i)$  be the optimal solution we can get from the values $\tF(\cell_{i},k)$ by placing $i$ squares.
By the convexity of $\tF(\cell_{i},k)$,
the following greedy algorithm is optimal:
as long as we still have budget, we assign 1 more square to the cell
which provides the largest increment of the objective value.
In fact, this greedy algorithm runs in almost the same way
as \GREEDYP. The only difference is that  \GREEDYP\ only picks
an entire edge of the convex hull, while the greedy algorithm here
may stop in the middle of an edge (only happen for the last edge).
Since the marginal increment never increases, we
can see that $\OPT_{\tF}(i)$ is concave.

By the way of choosing cells in our greedy algorithm,
we make the following  simple but important observation:
$$
\greedy(m')=\OPT_{\tF}(m')=\OPT_{\F}(m').
$$
So, our greedy algorithm is in fact optimal for $m'$.
%Since $\tF(\cell_{i},k)$ is always no smaller than $F(\cell_{i},k)$,
%we have $\OPT_{\tF}(m')\geq \OPT_{\F}(m')$.
Combining with $m-m'\leq\upperd$ and the concavity of $\OPT_{\tF}$,
we can see that
$$
\OPT_{\tF}(m')\geq
  \frac{m-\upperd}{m} \OPT_{\tF}(m)\geq 
  \left(1-\frac{\varepsilon^{2}}{9}\right)\OPT_{\tF}(m)\geq
  \left(1-\frac{\varepsilon^{2}}{9}\right)\OPT_{\F}(m).
$$
The last inequality holds because
$\OPT_{\tF}(i)\geq \OPT_{\F}(i)$ for any $i$.
The second last inequality holds because
$m>\frac{324}{\varepsilon^4}$ and $b=\frac{36}{\varepsilon^2}$.
\end{proof}

\subsection{Computing $\F(\cell,k)$ }
\label{subsec:F}
Now we show the subroutine $\GeneralmPartion$ for computing $\F(\cell,k)$.

We use a similar partition algorithm as Section~\ref{subsec:partition}.
The only difference is that this time we need to partition the cell finer
so that the maximum possible weight of points between any two adjacent parallel partition lines is $(\frac{\varepsilon^{3}W_{\cell}}{432})$.
After partitioning the cell, we enumerate all the possible ways of placing $k$
unit squares at the grid point.
Similarly, for each unit square $r$, we only count the weight of points that are in some cell fully covered by $r$.

We can adapt the algorithm in~\cite{de2009covering} to enumerate these possible choices in
 $O((\frac{1}{\varepsilon})^{\Delta_1})$
time where $\Delta_1=O(\sqrt{m_c})=O(\min(\sqrt{m},\frac{1}{\varepsilon}))$. The details are deferred to Subsection~\ref{subsec:Enum}.
Now we prove the correctness of this algorithm.
\begin{lemma}
$\GeneralmPartion$ returns a $(1-\frac{\varepsilon}{3})$
approximation to $\OPT(\cell_i,k)$.
\end{lemma}
\begin{proof}
We can use $(\frac{6}{\varepsilon})^{2}$ unit squares to cover the entire cell, so $\OPT(\cell_i,k)\geq \frac{k\varepsilon^{2}W_{\cell}}{36}$.
By the same argument as in Theorem~\ref{thm:main1},
the difference between $\OPT(\cell_i,k)$
and the answer we got are at most $4k$ times the maximum possible weight of points between two adjacent parallel partition lines.
Therefore, the algorithm returns a $(1-\frac{\varepsilon}{3})$-approximate answer of $\OPT(\cell_i,k)$.
\end{proof}

Now we can conclude the following theorem.
\begin{theorem}
\label{thm:main2}
Let $P$ be a set of $n$ weighted point, for any $0<\varepsilon<1$ we can find a $(1-\varepsilon)$-approximate answer
	for $\MaxRS(\calP, m)$
	in time
	$$
	O\left( \frac{n}{\varepsilon}\log \frac{1}{\varepsilon}+
    \frac{m}{\varepsilon}\log m+
    m\left(\frac{1}{\varepsilon}\right)^{\Delta_1}\right),
	$$
	where $\Delta_1=O(\min(\sqrt{m},\frac{1}{\varepsilon}))$.
\end{theorem}

\begin{algorithm}[h]
	\caption{$\MaxRS(\calP, m)$}
	\begin{algorithmic}[]
		\State $w_{\max}\leftarrow0$
		\For {each $G\in\{G_{\frac{6}{\varepsilon}}(0,0),...,G_{\frac{6}{\varepsilon}}(\frac{6}{\varepsilon}-1,\frac{6}{\varepsilon}-1)\}$}
		\State Use perfect hashing to find all the non-empty cells of $G$.
		\For {each non-empty cell $\cell$ of $G$}
		\State $r_{\cell}\leftarrow$ Algorithm~\ref{algo:mainalgo1} for $\cell$ with approximate ratio $(1-\frac{\varepsilon^{2}}{9})$
		\EndFor;
		\State Find the $m$ cells with the largest $r_{\cell}$.
		Suppose they are $\cell_{1},...,\cell_{m}$.
		\For{$i\leftarrow 1$ \textbf{to} $m$}
		\For{$k\leftarrow 1$ \textbf{to} $\upperd$}
		$\F(\cell_{i},k)\leftarrow$ \GeneralmPartion(c,k)
		\EndFor
		\EndFor;
		\State {\bf if} $m\leq324(\frac{1}{\varepsilon})^{4}$, {\bf then} $r\leftarrow \DP(\{\F(\cell_{i},k)\})$
        \State {\bf else}  $r\leftarrow \GREEDYP(\{\F(\cell_{i},k)\})$
		\State {\bf if} $w(r)> w_{\max}$, {\bf then} $w_{\max}\leftarrow w(r)$ and $r_{\max}\leftarrow r$
		\EndFor;
		\State \textbf{return} $r_{\max}$;
	\end{algorithmic}
    \label{algo:mainalgo2}
\end{algorithm}

\begin{proof}
The algorithm is summarized in Algorithm~\ref{algo:mainalgo2}.	
By Lemma~\ref{lm:greedyratio}, the greedy algorithm computes an $(1-\varepsilon^2/9)$-approximation to
$\OPT_{\F}(m)$.
Since $\F(\cell_i,k)$ is $(1-\frac{\varepsilon}{3})$-approximation to $\OPT(\cell_i,k)$, we get
$\OPT_{\F}(m)\geq (1-\frac{\varepsilon}{3}) \OPT(m)$.
By Lemma~\ref{selectm}, we get
$\OPT(m)\geq(1-\frac{\varepsilon^{2}}{9})\OPT$.
(Recall that $\OPT$ denotes the optimal solution that does not intersect $G$.)
Altogether, the greedy algorithm computes an $(1-\varepsilon^2/9)(1-\varepsilon^2/9)(1-\varepsilon/3)$ approximation to $\OPT$.
Moreover, by Lemma~\ref{shifting}, Algorithm~\ref{algo:mainalgo2} returns a $(1-\frac{2\varepsilon}{3})(1-\frac{\varepsilon^{2}}{9})(1-\frac{\varepsilon^{2}}{9})(1-\frac{\varepsilon}{3})$ approximation to the original problem.
Since $(1-\frac{2\varepsilon}{3})(1-\frac{\varepsilon^{2}}{9})(1-\frac{\varepsilon^{2}}{9})(1-\frac{\varepsilon}{3})>(1-\varepsilon)$, Algorithm~\ref{algo:mainalgo2} does return a $(1-\varepsilon)$-approximate solution.

\smallskip We now calculate the running time. Solving the values $f(\cell_i,1)$ and finding out the top $m$ results require $O(n\log \frac{1}{\varepsilon})$ time.
We compute the values $\F(\cell_i,k)$ of $m$ cells. For each cell $\cell_i$, we partition it only once and calculate $\F(\cell_i,1),\ldots,\F(\cell_i,b)$ using the same partition.
Computing the values $\F(\cell_i,k)$ of all $m$ cells requires
$O(n\log (\frac{1}{\varepsilon})+m(\frac{1}{\varepsilon})^{\Delta_1})$ time.
The greedy algorithm costs $O(m\log m)$ time.
We do the same for $\frac{6}{\varepsilon}$ different grids.
Therefore, the overall running time is as we state in the theorem.
\end{proof}

\subsection{Enumeration in \GeneralmPartion}
\label{subsec:Enum}

We can adapt the algorithm in~\cite{de2009covering} to enumerate these possible ways of placing $k$
unit squares at the grid point in $O((\frac{1}{\varepsilon})^{\Delta})$ time
where $\Delta=O(\sqrt k)$. We briefly sketch the algorithm.
We denote the optimal solution as $\OPT_\cell$.
From \cite{agarwal2002exact} we know that for any optimal solution,
there exists a line of integer grid (either horizontal or vertical) that
intersects with  $O(\sqrt{k})$ squares in $\OPT_\cell$, denoted as
the {\em parting line}.
So we can use dynamic programming.
At each stage, we enumerate the parting line, and the $O(\sqrt{k})$ squares
intersecting the parting line.
We also enumerate the number of squares in each side of the parting line in the optimal solution.
The total number of choices is $O((\frac{1}{\varepsilon})^{\Delta})$.
Then, we can solve recursively for each side. In the recursion,
we should consider a subproblem which is composed of a smaller rectangle, and an enumeration of  $O(\sqrt{k})$ squares of the optimal solution intersecting the boundary of the rectangle and
at most $k$ squares fully contained in the rectangle.
Overall, the dynamic programming
can be carried out in $O((\frac{1}{\varepsilon})^{\Delta})$
time.

\clearpage

\section{Extension to Other Shapes}\label{sec:othershape}

Our algorithm can easily be extended to solve other shapes.
We show the extension in this section.
The framework is almost the same as before.
The major difference is the way for building an $(1-\varepsilon)$-approximation in each cell
(the partition scheme in Section~\ref{subsec:F} works only for rectangles).

\subsection{Assumptions on the general shape}

Now, we assume that $\D$ is a shape subject to the following conditions.

\begin{enumerate}
\item[C-1] It is connected and closed, and its boundary is a simple closed curve.
%\item[C-2] The boundary of $\D$ comprises of $\tau$ bounded degree arcs,
%    where $\tau$ is a fixed constant.
%    By ``bounded degree'', we mean that there exists a constant %$\mathsf{deg}$ so that each arc on the boundary of $D$ is a polynomial %curve with degree less than or equal to $\mathsf{deg}$.
\item[C-2] It is contained in an axis-paralleled square of size $1\times 1$, and on the other hand it contains an axis-paralleled square of size $\sigma \times \sigma$, where $\sigma=\Omega(1)$. For convenience, we assume that $\frac{1}{\sigma}$ is an integer.
\item[C-3] Let $\partial \D$ denote the boundary of $\D$.
If we place $k$ copies of $\D$ in $\R^2$, the arrangement defined by
their boundaries contains at most $O(k^2)$ cells.
\end{enumerate}

%We also assume that the points in $\calP$ are in general position so that
%\begin{enumerate}
%\item[C*] For any two points $X_1,X_2$ in $\calP$, there are constant many %copies of $\D$ whose boundary pass through $X_1$ and $X_2$, and all these %copies can be enumerated in constant time.
%    (Note that C* naturally holds when: (1) $\D$ satisfies C-2 and (2) the %points in $\calP$ are in general position.)
%\end{enumerate}

%These assumptions are common to see. 
%For example, they are quite similar to those used in %\cite{agarwal2002translating}.

\vspace{0.1cm}
\noindent
{\bf Remark:} The above assumptions are quite general.
Now, we list some shapes  satisfying those assumptions.
\begin{enumerate}
	\item Disks and ellipsoid;
	\item Convex polygons with constant size
	(e.g., triangles, pentagons, hexagons).
	For a convex body $C$ in the plane, it is known that there is a rectangle $r$ inscribed in $C$ such that a homothetic copy $R$ of $r$ is circumscribed about $C$ and the positive homothetic ratio is at most $2$
	\cite{approx_rectangle}.
	Therefore, we can always affine-transform a convex body so that it satisfies C-2, with $\sigma = 1/2$.
	C-3 is also easy to see: in the arrangement defined by
	their boundaries, there are $O(k^2)$ intersection points or segments.
	Since the arrangement defines a planar graph, by 
	Eular's formula, there are $O(k^2)$ cells. 
	\item Following the same argument,
	we can also handle the case where
	$\D$ satisfies C-1 and C-2, and
	the boundary of $\D$ comprises of $\tau$ bounded degree arcs,
	where $\tau$ is a fixed constant.
	By ``bounded degree'', we mean that there exists a constant $\mathsf{deg}$ so that each arc on the boundary of $D$ is a polynomial curve with degree less than or equal to $\mathsf{deg}$.
\end{enumerate}

\medskip For convenience, we introduce some notation.
Let $\calU_b$ be the collection of sets that are the union of $b$ copies of $\D$.
In particular, $\calU_1=\{S \mid S \text{ is a translate of }\D\}$.
Let $\RD$ denote the shape constructed by rotating $\D$ by $\pi$,
namely, the only shape that is centrally-symmetric to $\D$.

\eat{
The following lemma states a simple fact about multiple copies of $\RD$.
It is later used in proving the shattering dimension of the range space $\calU_1$.

\begin{lemma}\label{lemma:arr-cells-bound}
Suppose that we have $m$ copies of $\RD$ arranged in the plane.
The boundaries of these bodies constitute to an arrangement.
We claim that the number of cells in this arrangement is bounded by $O(m^2)$.
\end{lemma}

\begin{proof}
Take two copies of $\RD$ and consider their respective boundaries.
Either of them comprise of $\tau$ bounded degree arcs according to assumption C-2.
Moreover, there are $O(1)$ intersections between any two bounded degree arcs.
So, there are $O(\tau^2)=O(1)$ intersections between the two boundaries.

In total, there are $O(m^2)$ intersecting points in the arrangement considered in this lemma.
Furthermore, applying the Euler's formula for planar graph,
  the number of faces and edges and hence the total number of cells in this arrangement is upper bounded by $O(m^2)$.
\end{proof}
}

\subsection{The shifting technique}

For the general shape, we consider grids with side length $s=6/ (\sigma^2\varepsilon)$.

Again for simplicity, we assume that $\frac{1}{\varepsilon}$ is an integer
and no point in $\calP$ has an integer coordinate.
We shift the grid to $s$ different positions: $(0,0),(1,1),....,(s-1,s-1)$.
Let $\mathbb{G}=\left\{G_{s}(0,0),...,G_{s}(s-1,s-1)\right\}.$

As we will see in the next lemma, the description of the shifting technique will be slightly more complicated than the original case for the squares. In the original case, for each grid $G$ in $\mathbb{G}$ we shift the $m$ squares so that no squares intersect with $G$. In the general case, we do \textbf{not} shift the shapes.
Instead, for each grid $G$, we ``assign'' each of the $m$ copies of $\D$ into one cell of $G$.
By assigning a copy to a cell $\cell$, we do not shift it to make it lie in $\cell$ (so, we do not require that this copy lies entirely inside $\cell$; it may intersect the boundary of $\cell$ and so intersects $G$).
When a copy $\D'$ is assigned to cell $\cell$ of $G$,
we assume that it only covers the points inside $\cell$.
The \emph{effective region} of $\D'$ is defined as $\D'\cap \cell$.

\begin{lemma}
\label{shifting-general}
There exist $G^\star\in\mathbb{G}$ such that we can place $m$ copies of $D$
and assign these copies to the cells of $G^\star$,
 so that the union of effective regions of these copies covers $(1-\frac{2}{3}\varepsilon) \times \MaxCRS(\calP,m)$
 weight of points.
\end{lemma}

An equivalent description is the following.

\begin{lemma}
For a grid $G$ in $\mathbb{G}$, let $\cell_1,\ldots,\cell_t$ denote the nonempty cells.
Define
$$
\OPT_G=\max \left(\sum_i \MaxCRS(\calP_{\cell_i},k_i) \mid \sum_i k_i=m\right).
$$
Then, $(\max_{G\in \mathbb{G}} \OPT_G)$ is an $(1-\frac{2}{3}\varepsilon)$-approximation of $\MaxCRS(\calP,m)$.
\end{lemma}

\begin{proof}[Proof of Lemma~\ref{shifting-general}]
The proof is similar as the that of Lemma~\ref{shifting}.

For any point $p$, we can always use $(2/\sigma)^2$ copies of $D$ to cover the $2\times2$ square centered at $p$.
Therefore, there exists an optimal solution $\mathsf{OPT}$ such that each covered point is cover by at most $(2/\sigma)^2$ copies in $\mathsf{OPT}$.

For each grid $G_s(i,i)\in\mathbb{G}$, we build a modified answer $R_{i}$ from $\mathsf{OPT}$ in the following way.
    For each copy $\D'$ of $\D$ that intersects with $G_s(i,i)$, there are two different situations.
    If $\D'$ only intersects with one vertical line or one horizontal line.
    We assign $\D'$ to one side of the line with bigger weight.
    In this case we will lose at most half of the weight of $\D'$.
    Notice that this kind of copies can only intersect with two grids in $\mathbb{G}$.
	Similarly, If $\D'$ intersects with one vertical line and one horizontal line at the same time,
    we assign it to one of the four quadrants derived by these two lines to keep the most weight.
    In this case we will lose at most 3/4 of the weight of $\D'$.
    This kind of copies can only intersect with one grid in $\mathbb{G}$.
    Now we calculate the sum of the weights we lose from $R_{0},R_{1},...,R_{s-1}$, which is at most $\max\{1/2\times2,3/4\times1\}=1$ times the value of $\mathsf{OPT}$.
    By the definition of $\mathsf{OPT}$, it is at most $(2/\sigma)^2 w(\mathsf{OPT})$.
    So the sum of the ``effective weights'' of $R_{0},R_{1},...,R_{s-1}$ is at least
	$(s-(2/\sigma)^2)\cdot w(\mathsf{OPT})$.
    The effective weight of $R_i$ is defined as the total weight covered by the union of the effective regions of $R_i$.
    Recall that $s=6/ (\sigma^2\varepsilon)$.
    By pigeon's principle, there exists some $i$ such that the effective weight of $R_{i}$ is at least $(1-\frac{2\varepsilon}{3})\cdot w(\mathsf{OPT})$.
\end{proof}

\subsection{Compute a $(1-\varepsilon)$-approximation to $\MaxCRS(\calP,m)$}

We give the framework in Algorithm~\ref{algo:mainalgo3}.

\begin{algorithm}[h]
	\caption{$\MaxCRS(\calP, m)$}
	\begin{algorithmic}[]
		\State $w_{\max}\leftarrow0$
		\For {each $G\in \mathbb{G}$}
		\State Use perfect hashing to find all the non-empty cells of $G$.
		\State \textbf{for} {each non-empty cell $\cell$ of $G$}
		\State \qquad $v_{\cell}\leftarrow$ $(1-\frac{\varepsilon^{2}}{9})$ approximation to $\MaxCRS(\Pc,1)$

		\State Find the $m$ cells with the largest $v_{\cell}$. Suppose they are $\cell_{1},...,\cell_{m}$.
        \State Let $b \leftarrow \min(m, \frac{s^2}{\sigma^2})=\min(m,\frac{36}{\sigma^6\varepsilon^2})$, which is the maximum number of copies put into a cell.
		\For{$i\leftarrow 1$ \textbf{to} $m$}
            \State Let $\cell$ denote $\cell_i$ for short.
    		\State \textbf{for} {$k\leftarrow 1$ \textbf{to} $\upperd$}
    		\State \qquad $\F(\cell_{i},k)\leftarrow$ $(1-\frac{\varepsilon}{3})$ approximation to $\MaxCRS(\Pc,k)$
		\EndFor;
		\State {\bf if} $m\leq \frac{324}{\sigma^6\varepsilon^4}$, ($\frac{324}{\sigma^6\epsilon^4}$ is chosen so that $\frac{m-b}{m}\geq (1-\varepsilon^2/9)$ and Lemma~6 remains true.)
        {\bf then} $r\leftarrow \DP(\{\F(\cell_{i},k)\})$ {\bf else} $r\leftarrow \GREEDYP(\{\F(\cell_{i},k)\})$
		\State {\bf if} $w(r)> w_{\max}$, {\bf then} $w_{\max}\leftarrow w(r)$ and $r_{\max}\leftarrow r$
		\EndFor;
		\State \textbf{return} $r_{\max}$;
	\end{algorithmic}
    \label{algo:mainalgo3}
\end{algorithm}

The correctness proof is exactly the same as the proof for Algorithm~\ref{algo:mainalgo2}.

Although, the framework is the same, the way for computing approximation of $\MaxCRS(\Pc,1)$ and $\MaxCRS(\Pc,k)$ is different from the square case, since the partition technique does not apply.
We show the new method in the next subsection and then analyze the
running time of Algorithm~\ref{algo:mainalgo3}.

\subsection{Compute a $(1-\varepsilon)$-approximation to $\MaxCRS(\calP_\cell,k)$}

\begin{definition}
For a weighted point set $\cal P$ and a range space $\cal U$ (which is a set of regions in the plane), we say another weighted point set $\calA$ is a $1/r$-approximation of $\calP$ with respect to $\calU$, if $\calA$ and $\calP$ have the same total weights and
$|w(\calA\cap U) - w(\calP \cap U)|<w(\calP)/r$ for any $u\in \calU$.
\end{definition}

The following lemma follows very similar argument in \cite{de2009covering}.

\begin{lemma}\label{lemma:ext_app_app}
Let $r_\varepsilon=72/(\varepsilon^3\sigma^6)$ and denote it by $r$ when $\varepsilon$ is clear.
Assume that $\Ac$ is a $1/r$-approximation of $\Pc$ with respect to $\calU_b$.
For $1\leq k\leq b$, if $\UA$ is an optimal solution for $\MaxCRS(\Ac, k)$,
then it is an $(1-\varepsilon)$-approximation to $\MaxCRS(\Pc, k)$.
\end{lemma}

\begin{proof}
Let $\UP$ denote the optimal solution for $\MaxCRS(\Pc, k)$.
Since $\UA$ is optimal for $\MaxCRS(\Ac, k)$, we have $w(\UA\cap \Ac)\geq w(\UP\cap \Ac)$.

Since $\cell$ is of size $s\times s$, and $\D$ contains an axis-paralleled square with size $\sigma\times \sigma$,
we have $\MaxCRS(\Pc,1)\geq \frac{\sigma^2}{s^2} w(\Pc)$. Recall that $s=\frac{6}{\sigma^2\varepsilon}$. So,
\[\frac{1}{r}w(\Pc)\leq \frac{s^2}{\sigma^2r} \MaxCRS(\Pc,1)=\frac{\varepsilon}{2}\MaxCRS(\Pc,1)\leq \frac{\varepsilon}{2}\MaxCRS(\Pc,k).\]

Since $\calA_\cell$ is a $1/r$-approximation of $\calP_\cell$, we have
\[\left| w(\UA\cap \Ac)-w(\UA\cap \Pc)\right| \leq w(\Pc)/r\leq \frac{\varepsilon}{2}\MaxCRS(\Pc,k)\]
and
\[\left| w(\UP\cap \Ac)-w(\UP\cap \Pc)\right| \leq w(\Pc)/r \leq \frac{\varepsilon}{2}\MaxCRS(\Pc,k).\]

Therefore
\[\begin{split}
w(\UA\cap \Pc)
 &=    w(\UA\cap \Ac)-w(\UA\cap \Ac) +w(\UA\cap \Pc)\\
 &\geq w(\UA\cap \Ac)-\frac{\varepsilon}{2}\MaxCRS(\Pc,k)\\
 &\geq w(\UP\cap \Ac)-\frac{\varepsilon}{2}\MaxCRS(\Pc,k)\\
 &=    w(\UP\cap \Pc)- w(\UP\cap \Pc)+w(\UP\cap \Ac)-\frac{\varepsilon}{2}\MaxCRS(\Pc,k)\\
 &\geq w(\UP\cap \Pc)-\frac{\varepsilon}{2}\MaxCRS(\Pc,k)-\frac{\varepsilon}{2}\MaxCRS(\Pc,k)\\
 &= (1-\varepsilon)\MaxCRS(\Pc,k)
\end{split}\]
This finishes the proof of the lemma.
\end{proof}

By this lemma, to compute an $(1-\varepsilon)$ approximation of $\MaxCRS(\Pc,k)$,
we can first build a $(1/r_\varepsilon)$-approximation $\Ac$ of $\Pc$ with respect to $\calU_b$ (for some $b\geq k$)
and then apply an exact algorithm to $\Ac$.

\medskip First, we show how to build a $1/r$-approximation $\Ac$ of $\Pc$.
We assume the reader is familiar with $1/r$-approximation for general range spaces.

\begin{definition}
For a range space $(X,\calU)$ with shattering dimension $d$,
we say that it admits a \emph{subspace oracle},
if given a set $Y\subseteq X$, a list of all distinct sets of the form
$Y \cap U$ for some $U\in \calU$ can be returned in $O(|Y|^{d+1})$ time.
\end{definition}

\begin{lemma}[\cite{matousek1995approx}]
\label{lemma:app}
Let $X$ be a weighted point set. Assume $(X,\calU)$ is a range space with shattering dimension $d$ and admits a subspace oracle.
For any parameter $r$, we can deterministically compute a $1/r$-approximation of size $O(r^2\log r)$ for $X$ with respect to $\calU$, in time $O(|X|\cdot (r^2\log r)^d)$.
\end{lemma}

%注：原文中r有<=n的限制。我并不清楚去掉后会不会有问题。但是deberg文章里没管这个限制。

\begin{lemma}\label{lemma:ext_app_size_time}
Suppose that $X$ is a set of weighted points and $r$ is a real.
\begin{enumerate}
\item[(1)] We can construct a $1/r$-approximation of $X$ with respect to $\calU_1$, of size $O(r^2\log r)$, in $O(r^4\log^2 r |X|)$ time.
\item[(2)] For an integer $b>1$, we can construct a $1/r$-approximation of $X$ with respect to $\calU_b$, of size $O((rb^2)^2 \log(rb^2))$, in $O((rb^2)^{12}\log^6(rb^2)|X|)$ time.
\end{enumerate}
\end{lemma}

\begin{proof}
(1) First of all, we claim that the range space $(X,\calU_1)$ has shattering dimension 2.
%We provide two proofs of this claim.
%Either one has its own merit.
%The first is more concise while the second is easier to follow.

%Proof I: As pointed out by \cite{Har-peled-GAA}, the shattering dimension of a range space defined by a family of shapes is
%  always bounded by the number of points that determine a shape in the family,
%  so $(X,\calU_1)$ has shattering dimension $2$.

%Proof II: 
  We designate a fixed special point in $\D$, call the pivot point
  of $\D$.
  For a point $A\in \R^2$, when we say ``we place a copy of $\D$
  at $A$'', it means the pivot point of the copy is placed at $A$. 
  Assume $X=\{X_1,\ldots, X_k\}$ is a set of $k$ points in $\R^2$.
  For each point $X_i$, we place a copy of $\RD$ at $X_i$
  (denoted it by $\RD_i$).
  In fact, if we place a copy of $\D$ such that its pivot point
  is in $\RD_i$, this copy of $\D$ can cover $X_i$.
  %  (More specifically, we make a copy of $\RD$ and translate it so that %its gravity center locates at $X_i$.)
  Let $\Gamma$ denote the arrangement of the boundaries of these $k$ copies of $\RD$.
%  Clearly, when (the gravity center of) a copy $\D'$ of $\D$ is restricted in a fixed cell of $\Gamma$,
%    the set $\{X \cap D'\}$ is consistent.
  By C-3, there are $O(k^2)$ cells in this arrangement.
  Placing a copy of $\D$ in any point of the same cell 
  cover the same subset of $X$.
  Therefore, the number of different subsets of $X$ that are shattered by $\calU_1$ is bounded by
  the number of cells of $\Gamma$.
  %On the other hand, By Lemma~\ref{lemma:arr-cells-bound}, the number of %cells in $\Gamma$ is bounded by $O(m^2)$.
  %Together, the number of subsets of $X$ shattered by $\calU_1$ is bounded %by $O(m^2)$,
  Hence, $(X,\calU_1)$ has shattering dimension 2.

\smallskip Next, we define a superset $\calU_1^*$ of $\calU_1$ and
  then construct a $1/r$-approximation with respect to $\calU_1^*$.
  Since $\calU_1^*$ is a superset of $\calU_1$, the approximation with respect to $\calU_1^*$ is
    also a $1/r$-approximation with respect to $\calU_1$, and thus we get (1).
  The simple reason we need to introduce $\calU^*_1$ is that it is much easier to construct a subspace oracle for $(X,\calU^*_1)$ instead of $(X,\calU_1)$.

  We define  $\calU_1^*$ as the union of $\calU_1$ and $\{x\mid x\in X\}$.
  So, each single point in $X$ constitutes a set in $\calU_1^*$.
  Note that $(X,\calU_1^*)$ has shattering dimension $2$.
    This immediately follows from the fact that $(X,\calU_1)$ has shattering dimension 2.

  We now construct a subspace oracle for $(X,\calU_1^*)$.
  Given a subset $Y\subseteq X$, we should return the sets in $\mathcal{S}=\{Y\cap U \mid U\in \calU_1^*\}$ in $|Y|^3$ time.
  Each set in $\mathcal{S}$ either contains a single point or contains at least two points.
    We can output those with a single point in $O(|Y|)$ time; they are exactly the single element subset of $Y$.
    We output other subsets in $\mathcal{S}$ as follows.
    Notice the following fact: If a copy of $\D$ contains at least two points in $Y$, we can shift this copy so that its boundary contains two points of $Y$, meanwhile its intersection with $Y$ is unchanged.
    So, we can enumerate two points $y_1,y_2$ in $Y$, and find all the copies of $\D$ whose boundary passes through $y_1,y_2$ and then output the points in $Y$ contained in each of such copy.

  To sum up, $(X,\calU_1^*)$ has shattering dimension 2 and admits a subspace oracle. Therefore, by applying Lemma~\ref{lemma:app}, we can construct a $1/r$-approximation of $X$ with respect to $\calU^*_1$, of size $O(r^2\log r)$, in $O(r^4\log^2 r |X|)$ time.

\medskip \noindent (2)
Let $\mathcal{V}$ be the infinite set of cells that can arise in a vertical decomposition (see Chapter~6 in \cite{Har-peled-GAA} or Chapter~6 in \cite{CGA1997}) of any collection of copies of $\D$ in the plane.
It is easy to see that $(X,\mathcal{V})$ has shattering dimension 6 (as proved in Lemma~3 in \cite{de2009covering}).
Briefly speaking, a cell is characterized by 6 points: the leftmost point, the rightmost point,
two points on the upper boundary of this cell, and two points on the lower boundary of this cell.
It follows that there is a subspace oracle for range space $(X,\mathcal{V})$.

Let $r'=v\cdot r$, where $v$ is the maximum number of cells that the vertical decomposition of $b$ copies of $\D$  can have.
According to the assumption of $\D$, given two different placements of $\D$, their boundaries have constant many intersections.
This implies that any arrangement of $b$ copies of $\D$ have $O(b^2)$ cells, which further implies that $v=O(b^2)$.
To construct a $1/r$-approximation for $X$ with respect to $\calU_b$, we apply the following fact proved in \cite{de2009covering}:
Let $A$ be a $1/(r')$-approximation for $X$ with respect to the ranges $\mathcal{V}$,
then $A$ is a $1/r$-approximation for $X$ with respect to $\calU_b$.

According to Lemma~\ref{lemma:app}, we can construct a $1/(r')$-approximation for $X$ of size $O(r'^2\log r')$, in $O(((r')^{2}\log(r'))^6|X|)$ time. Thus we get (2).
\end{proof}

Next, we show an exact algorithm for computing $\MaxCRS(X,m)$.

\begin{lemma}\label{lemma:ext_exact}
Assume that $X$ is a set of weighted points.
We can compute the exact solution to $\MaxCRS(X,m)$ in $O(|X|^{2m+1})$ time.
%In addition, when $m=1$, we can compute the exact solution to $\MaxCRS(X,m)$ in $O(|X|^2)$ time.
%we can compute the exact solution for $\calA$ in $O((\frac{k}{\varepsilon})^{O(\sqrt{k})})$ time.
\end{lemma}

\begin{proof}
In the optimal solution of $\MaxCRS(X,m)$, we can always choose those copies of $\D$ so that each of them contains a single point or contains two points on their boundary. We call these copies the critical copies.
According to our assumption, the number of critical copies is $O(|X|^2)$, and these copies can be enumerated
in $O(|X|^2)$ time.
For each critical copy $\D_i$, we compute and store in memory the list of points $L_i$ that covered by $\D_i$.
Then, to compute the exact solution of $\MaxCRS(X,m)$, we enumerate all possible combination of $m$ critical copies and find the optimum one.
Since the points covered by each copy is stored, the points covered by the union of these copies can be computed in $O(|X|)$ time. Thus the running time is $O(|X|^{2m+1})$.
\footnote{The running time can be improved to $O(|X|^2)$ time for $m=1$. For $m=1$, we do not need to store the list $L_i$ for each $i$, and we can compute the summation of the weights in $L_i$ in amortized $O(1)$ time. (See \cite{agarwal2002translating} or \cite{de2009covering} for the details). We do not apply this optimization.}
\end{proof}

\subsection{Analysis of running time}

Now we provide some details and the running time analysis of Algorithm~\ref{algo:mainalgo3}.

Recall that $r_\varepsilon = O(\varepsilon^{-3})$.
For each non-empty cell $\cell$, we compute $\Ac$, which is a $1/r_{\varepsilon^2/9}$-approximation of
$\Pc$ with respect to $\calU_1$. This costs
\[O\left(n_cr_{\varepsilon^2/9}^4\log^2(r_{\varepsilon^2/9})\right)=O(n_c\varepsilon^{-25})\]
time according to Lemma~\ref{lemma:ext_app_size_time}~(1);
and the size of $\Ac$ is
\[r_{\varepsilon^2/9}^2\log(r_{\varepsilon^2/9})=O(\varepsilon^{-13}).\]

Then, we compute $\MaxCRS(\Ac,1)$ and use it as the $(1-\varepsilon^2/9)$ approximation of $\MaxCRS(\Pc,1)$.
This costs
\[O(|\Ac|^3)=O(\varepsilon^{-39})\]
time according to Lemma~\ref{lemma:ext_exact}.
So the first inside loop costs $O(n\varepsilon^{-39})$ time.

For each cell $\cell$ in $\cell_1,\ldots,\cell_m$, we compute a $1/r_{\varepsilon/3}$-approximation $\Ac'$ of
$\Pc$ with respect to $\calU_b$. Recall that $b=O(\varepsilon^{-2})$. The time for computing $\Ac'$ is
\[O\left(n_c\times r_{\varepsilon/3}^{12}b^{24}\log^6(r_{\varepsilon/3}b^2)\right)=O(n_c\varepsilon^{-36-48-1})=O(n_c(\frac{1}{\varepsilon})^{O(1)})\]
according to Lemma~\ref{lemma:ext_app_size_time}~(2).
Therefore, the total time for computing $\Ac'$ for $c\in \{\cell_1,...,\cell_m\}$ is $O((\frac{1}{\varepsilon})^{O(1)}n)$.
Note that $\Ac'$ is of size
\[O\left(r_{\varepsilon/3}^2b^4\log(r_{\varepsilon/3}^2b^4)\right)=O(\varepsilon^{-15}).\]
Therefore, the total time for computing $\{F(c_i,k): 1\leq i\leq m, 1\leq k\leq b\}$
is $O(m(\varepsilon^{-15})^{2b+1})$ according to Lemma~\ref{lemma:ext_exact}.

The overall running time is $O(n(\frac{1}{\varepsilon})^{O(1)}+\frac{m}{\epsilon}\log m + m(\frac{1}{\varepsilon})^{\Delta_2})$,
where $\Delta_2=O(\min(m,\frac{1}{\varepsilon^2}))$.
The second term comes from the greedy algorithm.

\section*{Reference}

\bibliographystyle{elsarticle-num}
\bibliography{TCS-cover}

\appendix

\section{The values $\F(\cell_{i},k)$ may not be concave}

\begin{figure}[h]
	\centering
	\includegraphics[width=0.8\textwidth]{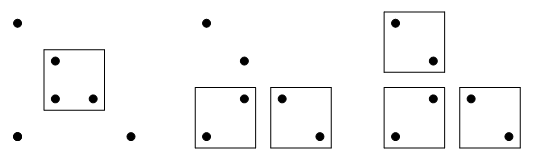}
	\caption{$\F(\cell_{i},k)$ may not be concave: $\F(\cell_{i},1)=3$, $\F(\cell_{i},2)=4$, $\F(\cell_{i},3)=6$.}
	\label{fig:convex}
\end{figure}

\end{document}